\documentclass{article}

\usepackage{PRIMEarxiv}

\usepackage[utf8]{inputenc} 
\usepackage[T1]{fontenc}    
\usepackage{hyperref}       
\usepackage{url}            
\usepackage{booktabs}       
\usepackage{amsfonts}       
\usepackage{nicefrac}       
\usepackage{microtype}      
\usepackage{lipsum}
\usepackage{fancyhdr}       
\usepackage{graphicx}       
\usepackage{blkarray}
\usepackage{amsmath}
\usepackage{algorithmic}
\usepackage[linesnumbered,vlined,boxed,commentsnumbered,ruled,norelsize]{algorithm2e}
\usepackage[table]{xcolor}
\usepackage{multirow}
\usepackage{subcaption}
\usepackage{amsthm}
\newtheorem{theorem}{Theorem}
\newtheorem{example}[theorem]{Example}
\newtheorem{definition}[theorem]{Definition}

\title{Liars are more influential: Effect of Deception in Influence Maximization on Social Networks}

\author{
  Mehmet Emin Aktas\\
  Department of Mathematics and Statistics\\
  University of Central Oklahoma\\
  Edmond, OK 73034 \\
  \texttt{maktas@uco.edu} \\
  \And
 Esra Akbas \\
  Department of Computer Science\\
  Oklahoma State University\\
  Stillwater, OK 74078 \\
  \texttt{eakbas@okstate.edu} \\
  \And
   Ashley Hahn \\
  Department of Global and Sociocultural Studies\\
  Florida International University\\
  Miami, FL 33199 \\
  \texttt{ashley4hahn@gmail.com} \\
}

\begin{document}
\maketitle

\begin{abstract}
Detecting influential users, called the influence maximization problem on social networks, is an important graph mining problem with many diverse applications such as information propagation, market advertising, and rumor controlling. There are many studies in the literature for influential users detection problem in social networks. Although the current methods are successfully used in many different applications, they assume that users are honest with each other and ignore the role of deception on social networks. On the other hand, deception appears to be surprisingly common among humans within social networks. In this paper, we study the effect of deception in influence maximization on social networks. We first model deception in social networks. Then, we model the opinion dynamics on these networks taking the deception into consideration thanks to a recent opinion dynamics model via sheaf Laplacian. We then extend two influential node detection methods, namely Laplacian centrality and DFF centrality, for the sheaf Laplacian to measure the effect of deception in influence maximization. Our experimental results on synthetic and real-world networks suggest that liars are more influential than honest users in social networks.
\end{abstract}


\section{Introduction}

In understanding social life, social scientists examine social structure and interactions to shed light on the way they guide and are guided by human behavior. As Castell \cite{castells2000toward} suggests, if we see social structure as a collection of dynamic networks, then studying social networks would help us understand social life and human behavior. Network studies from anthropology and sociology demonstrated that individuals’ beliefs, behaviors, attitudes, and opinions can shape and be shaped by their social connections and hence their social networks~\cite{centola2010spread,smith2008social}. As one’s beliefs and opinions emerge in discourse dialogically \cite{keane2008evidence, lambek2010ordinary}, it is important to examine the network influence through focusing on social interaction and social actors in a network. However, not everyone in a network might be contributing to the network effect equally. Some individuals would be more influential than others in a network in terms of shaping others’ beliefs and opinions. Thus, detecting these more influential users, called the influence maximization problem on social networks, is not only an important graph mining problem with many diverse applications such as information propagation, market advertising, and rumor controlling but also an important social phenomenon that requires researchers' attention in understanding human behavior and social life. 



There are many studies in the literature for the influential node detection problem in networks. While some studies are based on degree of nodes such as degree centrality~\cite{bonacich1972factoring} and H-index~\cite{lu2016h}, some use paths in networks such as closeness centrality~\cite{freeman1978centrality} and betweenness centrality~\cite{freeman1977set}. Others use eigenvectors of graphs such as PageRank \cite{brin1998anatomy}. There are other studies that employ diffusion models on networks via graph Laplacian to find the influential nodes. For example, in \cite{qi2012laplacian}, the authors define the Laplacian energy of the network with the spectrum of the graph Laplacian. Then, the centrality of a node, called the Laplacian centrality, is measured as the drop of Laplacian energy in the network when that node and its adjacent edges are removed. As another example, in \cite{aktas2021influential}, the authors introduce the DFF centrality using the diffusion Frechet function, the weighted sum of the diffusion distance between a vertex and the rest of the network.



On the other hand, although the current node centrality methods are widely used in many different networks, including social networks, they suffer from an important issue, which is specifically crucial in social network settings: They assume that users are honest with each other and ignore the role of deception on social networks. However, deception appears to be surprisingly common among humans and within social networks. Social scientists have demonstrated that lying is a widely spread practice in societies and across cultures \cite{barnes1994pack,brown2002everyone}. Brown \cite{brown2002everyone} showed that people lie about their "private affairs, money, comings, and goings, and not having the things that people are asking to borrow." According to \cite{serota2010prevalence}, Americans have 1.65 lies per day although 23\% of all lies told by 1\% of individuals in the study. Other studies also reported that the average lie per person is between 0.6-2.0 \cite{depaulo1996lying,abeler2012truth} and lies are being less common in face-to-face interactions than in online interactions \cite{hancock2004deception}. Hence, deception is inevitable in social networks and one should take deception into consideration while modeling the information diffusion and detecting influential nodes in social networks.

Furthermore, not all lies are the same. Researchers classify human deception into four types: (1) \textit{prosocial}, lying to protect someone or to benefit or help others. For example, one can say to a minor who is learning to play violin that he is playing great although he may not; (2) \textit{self-enhancement}, lying to avoid embarrassment, disapproval, or punishment; (3) \textit{selfish}, lying to protect oneself at the expense of hurting others; (4) \textit{antisocial}, lying to hurt someone else intentionally. Even though there are different types of lies, all lies are socially and dialogically constructed and produce social effects in society. While some lies can be seen as morally acceptable, others can be seen as immoral. Consequently, different types of lies might have different effects on social networks. For example, antisocial lies might destroy relations since they are selfish whereas prosocial lies can keep relations in good conditions \cite{nyberg1994varnished,depaulo1998everyday}.

Effects of deception in social networks have also taken researchers' attention recently. For example, in \cite{iniguez2014effects}, the authors show how lying can cause social networks to become fragmented. They also study the effects of prosocial and antisocial lies separately. Furthermore, in \cite{barrio2015dynamics}, they find that lies shape the topology of social networks and cause the formation of tightly linked, small communities. They also find that liars are the ones that connect communities of different opinions, hence they have substantial centrality in the network. There are also many studies to detect deception on social networks. \cite{spottswood2016positivity} uses the positivity bias for predicting the use of prosocial lies on Facebook. This paper \cite{giatsoglou2015retweeting} studies the retweeting activity on Twitter to detect the deception. 

In this paper, we study the impact of lying in influence maximization on social networks. Our goal here is to understand whether lying makes people more influential in social networks. To reach our goal, we first model deception on social networks consisting of honest interactions and exchange of prosocial or antisocial lies between individuals motivating from \cite{iniguez2014effects}. We also define the honesty level of users with an honesty parameter to see the effect of different deception levels. Next, we model the opinion dynamics in social networks using the sheaf Laplacian \cite{hansen2021opinion}. Sheaf Laplacian provides a very flexible model that allows users to express their opinion however they choose and selectively lie to their neighbors. Next, we extend two node centrality measures, Laplacian centrality, and DFF centrality, for the sheaf Laplacian to detect influential nodes when the deception is present in the network. We prove that these centralities, which are originally defined for the graph Laplacian, can also be defined for the sheaf Laplacian. Then, we employ these centrality measures to detect the influence of each node to see the effect of deception in influence maximization. We repeat this process on synthetics and various real-world social networks. Our results show that liars, regardless of being prosocial or antisocial liars, are more influential than honest users.

The paper is formatted as follows. In Section \ref{sec:prelim}, we discuss the preliminary concepts for graphs, graph Laplacian and sheaf Laplacian. In Section \ref{sec:method}, we present our methodology on modeling deception in social networks and constructing the sheaf Laplacian. We also explain how we extend Laplacian and DFF centralities to the sheaf Laplacian. In Section \ref{sec:exp}, we explain our evaluation method and present our results on various synthetic and real-world datasets. Our final remarks with future work directions are found in Section \ref{sec:conc}.
\vspace{-0.1in}

\section{Preliminaries} \label{sec:prelim}
In this section, we first discuss the preliminary concepts for graphs and graph Laplacian. Next, we present the sheaf data structure and the sheaf Laplacian that we use to model the opinion dynamics on social networks when deception is present. 

\subsection{Graphs}
\textit{Graphs} are structured data representing relationships between objects \cite{aggarwal2010managing,Cook2006}. In a formal definition, a network $G$ is a pair of sets $G = (V, E)$ where $V$ is the set of vertices and $E \subset V \times V $ is the set of edges that  connections between pairs of vertices. If there is a score for the relationship between vertices that could represent the strength of interaction, we can represent this type of relationships or interactions by a \textit{weighted network}. In a weighted network, a weight function $W: E \rightarrow \mathbb{R}$ is defined to assign a weight for each edge.
Let $G$ be a weighted undirected graph with the vertex set $V$ and a weight function $w:V \times V \rightarrow \mathbb{R}_{\geq 0}$. The \textit{adjacency matrix} $A$ of $G$ is defined as the $n \times n$ matrix with $A(i,j)=w(v_i,v_j)$ for $i,j \in \{1,...,n\}$ with $n$ being the number of vertices of $G$. 

The graph Laplacian, first appeared in \cite{kirchhoff1847ueber} where the author analyzed flows in electrical networks, is an operator on a real-valued function on vertices of a graph. Let $D$ be the $n \times n$ diagonal matrix with $D(i,i)=\sum_j w(i,j)$, i.e., the weighted degree of the vertex $i \leq n$. We can define the graph Laplacian $L$ as $L=D-A$ where $D$ is the weighted degree matrix and $A$ is the weighted adjacency matrix. The graph Laplacian can be also found using the signed incidence matrix $B$. The \textit{incidence matrix} $B$ of $G$ is defined as the $n \times m$ matrix with $B(i,j)=1$ if $v_i \in e_j$ and 0 otherwise. One can also incorporate edge weight and edge orientations in $B$. Then, the graph Laplacian is given by $L=BB^T$. 

\subsection{Sheaf and Sheaf Laplacian}
A sheaf is a data structure associating data spaces to vertices and edges of a graph, with further telling how the data over different parts of the graph should be related. More formally, we can define a sheaf as follows.

\begin{definition}
For a given graph,$G=(V,E)$, a sheaf $\mathcal{F}$ on $G$ consists of a vector space $\mathcal{F}_v$ for each vertex $v \in V$, a vector space $\mathcal{F}_e$ for each edge $e \in E$, and a linear transformation $\mathcal{F}_{v \rightarrow e}: \mathcal{F}_v \rightarrow \mathcal{F}_e$ for each incident vertex-edge pair. 
\end{definition}

In the notion of modeling opinion dynamics on social networks, the vector space $\mathcal{F}_v$ over each vertex $v \in V$ is the opinion space of the vertex. Mathematically, this is a real vector space with a basis of the collection of topics where the basis consists of users' social, demographic, and cultural dynamics and the moral assemblages/bases behind their opinions. The scalar values on each basis element correspond to negative, neutral, or positive opinions about the topic. We would like to note here that these bases elements are not necessarily the same nor the same number for each user. For example, let the topic of discussion be the mask mandate during a pandemic. While for one user, the basis behind his opinion could be politics and religion, for another user, it could be health and isolation. They share their private opinion publicly based on their private opinion bases.

\begin{figure}[h!]
\centering
     \includegraphics[width=0.55\textwidth]{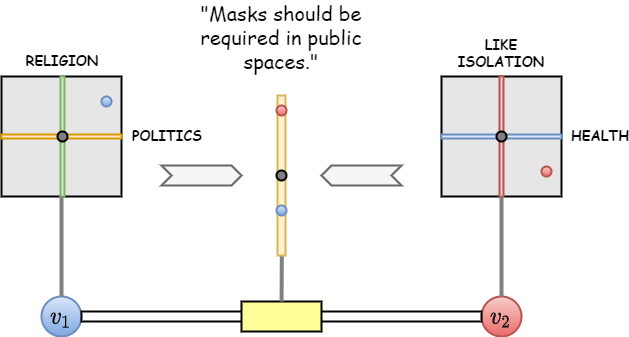}
    \caption{A sheaf structure on an edge} 
    \label{fig:sheaf}
   \end{figure}

See Figure \ref{fig:sheaf} for an illustration of this example. The vertex on the left ($v_1$) thinks both politics and religion are important in his opinion about the mask mandate and he does not support it. On the other hand, the vertex on the right ($v_2$) takes health very important but he also does not like the isolation. As a result, he does support the mask mandate.

Furthermore, in this opinion dynamics model, the vector space $\mathcal{F}_e$ over each edge $e \in E$ is the discourse space where each users represent their opinions on the topics of discussion by formulating stances as a linear combination of existing opinions on personal opinion basis. These expression of opinions are modeled using the linear transformations $\mathcal{F}_{v \rightarrow e}: \mathcal{F}_v \rightarrow \mathcal{F}_e$. For example, let $u$ and $v$ be two users that are connected with an edge $e$ in a social network. Let $x_u \in \mathcal{F}_u$ and $x_v \in \mathcal{F}_v$ be their opinions. If $\mathcal{F}_{u \rightarrow e}(x_u)=\mathcal{F}_{v \rightarrow e}(x_v)$, then there is a local consensus between $u$ and $v$. For example, in Figure \ref{fig:sheaf}, the users do not have a consensus initially since their public discourse on mask mandate does not coincide.

The sheaf Laplacian is defined similarly as the graph Laplacian. Let bundle all the data over vertices and over edges into a grouped vector spaces as follows \[C^0(G;\mathcal{F})=\bigoplus_{v \in V(G)} \mathcal{F}_v\]
\[C^1(G;\mathcal{F})=\bigoplus_{e \in E(G)} \mathcal{F}_e.\]
$C^0$ is called $0-cochains$ and it consists of a choice of data, $x_v \in \mathcal{F}_v$, for every vertex $v \in V$. Similarly, $C^1$ is called $1-cochains$ and it consists of a choice of data over each edge $e \in E$. Then, we tie the data over vertices (0-cochains) and edges (1-cochains) together with a linear transformation, called the \textit{\textbf{coboundary map}}, $\delta:C^0(G;\mathcal{F}) \rightarrow C^1(G;\mathcal{F})$. For an (arbitrarily) oriented edge $e=u \rightarrow v$, we define $\delta$ explicitly as follows:
\[
(\delta x)_e=\mathcal{F}_{v \rightarrow e}(x_v)-\mathcal{F}_{v \rightarrow e}(x_u).
\]
Then, the \textbf{sheaf Laplacian} is given by 
\[
L_{\mathcal{F}}=\delta^T\delta: C^0(G;\mathcal{F}) \rightarrow C^0(G;\mathcal{F}).
\]
The sheaf Laplacian does not depend on the choice of orientations while constructing the coboundary map. 

\begin{example}
The sheaf in Figure \ref{fig:sheafLap} has the following coboundary map
\[
\delta=\left[\begin{array}{c|cc|c|cc}
 -1 & -2 & 1 & 0 & 0 & 0 \\ \hline
 0 & -2 & 3 & -1 & 0 & 0 \\ \hline
 0 & 0 & 0 & 3 & -1 & 1 \\ \hline
 2 & 0 & 0 & 0 & -1 & 0 \\ 

\end{array}\right]\]
and the sheaf Laplacian
\[
L_{\mathcal{F}}=\left[\begin{array}{c|cc|c|cc}
 5 & 2 & -1 & 0 & -2 & 0 \\ \hline
 2 & 8 & -8 & 2 & 0 & 0 \\ 
 -1 & -8 & 10 & -3 & 0 & 0 \\ \hline
 0 & 2 & -3 & 10 & -3 & 3 \\ \hline
 -2 & 0 & 0 & -3 & 2 & -1 \\ 
  0 & 0 & 0 & 3 & -1 & 1 \\ 
\end{array}\right].\]

\begin{figure}[h!]
\centering
     \includegraphics[width=0.4\textwidth]{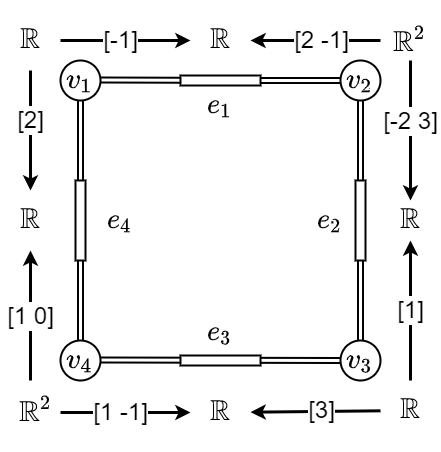}
    \caption{A sheaf structure on a 4-cycle graph. The dimension of vector spaces over $v_1$ and $v_3$ are 1 where it is 2 for $v_2$ and $v_4$. The dimension of vector spaces over edges are 1 as well.} 
    \label{fig:sheafLap}
   \end{figure}
  
\end{example}


\section{Methodology} \label{sec:method}
In this section, we first explain how we model deception in social networks inspiring from Iniguez et. al. \cite{iniguez2014effects}. Next, we discuss how we construct the sheaf Laplacian using the deception model in the previous step. Finally, we extend the centrality method originally defined for the graph Laplacian to the sheaf Laplacian to detect influential nodes in the network when deception is present. 

\subsection{Modelling deception in social networks}
For a vertex $v_i \in G$, let $x_i(t)$ represent the opinion of $v_i$ about a topic at time $t$. We can take $x_i\in [-1,1]$ with -1 meaning total disagreement and 1 meaning total agreement. These are the private opinion of users. To find the sheaf Laplacian of the graph, we need to know how each user discloses his opinion publicly. For the disclosing process, in this paper, we assume users are categorized into three groups: honest, prosocial liar and antisocial liar following \cite{iniguez2014effects}. As explained in the introduction, prosocial lies are said to benefit someone where antisocial lies are intended to hurt.

Inspiring from Iniguez et. al. \cite{iniguez2014effects}, we model these three different opinion disclosure with respect to users opinion. The amount of the information, $w_{ji}$,  flowing from $i$ to $j$ can be defined as
\begin{equation} \label{equ:flow}
w_{ji}=\begin{cases} 
      x_i & \text{ if user $i$ is honest} \\
      \tau x_i + (1-\tau)x_j & \text{ if user $i$ is prosocial liar} \\
      \tau x_i - (1-\tau)x_j & \text{ if user $i$ is antisocial liar}
   \end{cases}
\end{equation} 
where $\tau \in [0,1]$ is the honesty parameter. When $\tau=1$, liars are also honest and $\tau=0$, they are completely dishonest. 

Although for honest users, what others think does not change how they disclose their opinion, liars (both prosocial and antisocial) express their opinion based on the private opinion of their neighbors, i.e., $x_j$ in $w_{ji}$. This is an issue since users cannot know the private opinions of their neighbors, instead, they can only know how they disclose their opinion publicly, i.e., their public opinion. To tackle this issue, instead of using private opinion, we define the public opinion of the user $i$, $y_i$, with taking the average amount of the information flowing from this user to his neighbors as follows
\begin{equation} \label{equ:public}
\displaystyle y_i=\frac{1}{k_i}\sum_{j \in N_i} w_{ji}
\end{equation}
where $k_i$ is the degree of vertex $v_i$ and $N_i$ is the set of neighbors of $v_i$ in $G$. 

\subsection{Sheaf Laplacian construction}
The key information we need to construct the sheaf Laplacian $L_{\mathcal{F}}$ is the linear transformations $\mathcal{F}_{v \rightarrow e}: \mathcal{F}_v \rightarrow \mathcal{F}_e$ between a vertex and its neighbors. In other words, we need to know how each user discloses his opinion publicly with his neighbors using his opinion basis. As we discuss in the previous section, opinion disclosures (i.e., linear transformations) depend on whether the user is honest, a prosocial liar, or an antisocial liar. Based on the model in the previous section, we combine Equation \ref{equ:flow} and Equation \ref{equ:public} and obtain the linear transformation from $v_i$ to $v_j$ through the edge $e$ as follows.

\begin{equation}
\mathcal{F}_{v \rightarrow e}(x_i)=\begin{cases} 
      x_i & \text{ if user $i$ is honest} \\
      \tau x_i + (1-\tau)y_j & \text{ if user $i$ is prosocial liar} \\
      \tau x_i - (1-\tau)y_j & \text{ if user $i$ is antisocial liar}.
   \end{cases}
\end{equation}

An illustrative example is available in Figure \ref{fig:sheafHonest}. In the matrix notation, the linear transformation is given as 
\begin{equation}
\displaystyle
\mathcal{F}_{v \rightarrow e}:\begin{cases} \displaystyle
      [1] & \text{ if user $i$ is honest} \\
       [\tau + (1-\tau)\frac{y_j}{x_i}] & \text{ if user $i$ is prosocial liar} \\
       [\tau - (1-\tau)\frac{y_j}{x_i}] & \text{ if user $i$ is antisocial liar}.
   \end{cases}
\end{equation}

\begin{figure}[h!]
\centering
     \includegraphics[width=0.6\textwidth]{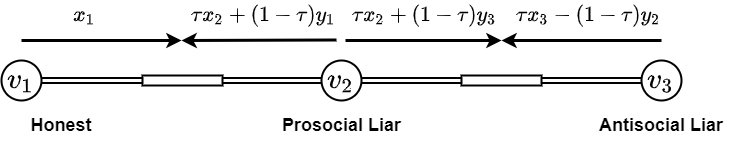}
    \caption{Linear transformations between vertices in the presence of a honest, a prosocial liar and an antisocial liar.} 
    \label{fig:sheafHonest}
   \end{figure}

In this model, we take the opinion space over each vertex and discourse space over each edge as 1-dimensional for simplicity. The opinion space over each vertex simply takes the private opinion $x_i$ for each user $i$, and the discourse space over each edge $e=v_i \rightarrow v_j$ takes the public disclosure of $i$th user's private opinion, $x_i$, based on the private opinion of $j$th user. On the other hand, this model can be generalized to any dimension of opinion and disclosure spaces.  

Next, to construct the sheaf Laplacian, we need to define the coboundary map $\delta$ for each edge using the linear transformations. Let $e=v_i \rightarrow v_j$ be an oriented edge. Then, the coboundary map on $e$ is defined as 
\begin{equation}
(\delta x)_e=\mathcal{F}_{v \rightarrow e}(x_i)-\mathcal{F}_{v \rightarrow e}(x_j).
\end{equation}
We outline the steps on constructing the sheaf Laplacian in Algorithm \ref{alg:spar}.

\begin{algorithm}[h]
\DontPrintSemicolon
\SetAlgoLined
\small
\KwIn{$G(V,E)$, graph, $X$, opinion distribution over $V$, $R$, the relation type of $V$ (honest, prosocial liar, antisocial liar), $\tau$, honesty parameter.}
\KwOut{$L_{\mathcal{F}}$, the sheaf Laplacian}

$P \leftarrow$ perception of users\;
$ind=0$\;
\ForEach{$v \in V$} {
    $p_v=$ PublicOpinion$(v,X,\tau, G, R_v)$\;
    $P(ind)=p_v$\;
    $ind=ind+1$\;
}

$B \leftarrow$ incidence matrix\;
$ind=0$\;
\ForEach{$e=(u,v) \in E$} {
    $d_u =$ Disclosure$(u,v,X,\tau, P, R_u)$\;
    $d_v =$ Disclosure$(v,u,X,\tau, P, R_v)$\; 
    $B(ind,u)=d_u$\;
    $B(ind,v)=-d_v$\;
    $ind=ind+1$\;
}
$L_\mathcal{F}=B^T*B$\;
\SetKwFunction{FMain}{Disclosure}
    \SetKwProg{Fn}{Function}{:}{}
    \Fn{\FMain{$u,v,X,\tau, P, R_u$}}{
        \uIf{$R_u$=honest}{
    $d_u=1$ \;
  }
  \uElseIf{$R_u$=prosocial liar}{
    $d_u=\tau + (1-\tau)*[P_v/X_u]$ \;
  }
  \Else{
    $d_u=\tau - (1-\tau)*[P_v/X_u]$ \;
  }
        \textbf{return} $ d_u; $ 
}

\SetKwFunction{FMain}{PublicOpinion}
    \SetKwProg{Fn}{Function}{:}{}
    \Fn{\FMain{$u,X,\tau, G, R_u$}}{
        $n_u \leftarrow$ neighbors of $u$ in $G$\;
        public=[]\;
        $ind=0$\;
        \ForEach{$v \in n_u$} {
            $d_u \leftarrow$ InfoAmount$(u,v,X,\tau, R_u)$\;
            public$(ind)=d_u$\;
            $ind=ind+1$\;
}
        $y_u=$average(public)\;
        \textbf{return} $ y_u; $ 
}

\SetKwFunction{FMain}{InfoAmount}
    \SetKwProg{Fn}{Function}{:}{}
    \Fn{\FMain{$u,v,X,\tau, R_u$}}{
        \uIf{$R_u$=honest}{
    $w=X_u$ \;
  }
  \uElseIf{$R_u$=prosocial liar}{
    $w=\tau*X_u + (1-\tau)*X_v$ \;
  }
  \Else{
    $w=\tau*X_u - (1-\tau)*X_v$ \;
  }
        \textbf{return} $ w; $ 
}

\caption{Sheaf Laplacian $L_{\mathcal{F}}$}
\label{alg:spar}

\end{algorithm}

\subsection{Influential node detection when deception is present}
The sheaf Laplacian allows us to model the information diffusion when deception is present. As the next step of this research, we need to detect the influential nodes in the network with deception. There are two influential node detection methods in the literature that use the graph Laplacian as the input, namely Laplacian centrality \cite{qi2012laplacian} and DFF centrality \cite{aktas2021influential}. Here, we extend these two centralities to the sheaf Laplacian.

\subsubsection{Laplacian Centrality}
This node centrality is based on the Laplacian energy of the network. The centrality of a node is measured as the drop of Laplacian energy in the network when that node and its adjacent edges are removed. The Laplacian energy is defined as follows.

\begin{definition}
Let $G$ be a weighted network on $n$ vertices and $L$ be the graph Laplacian of $G$ with the eigenvalues $\lambda_1, ..., \lambda_n$. Then, the Laplacian Energy of $G$ is given by
\[
\displaystyle E_L(G)=\sum_{i=1}^n \lambda_i^2.
\]
\end{definition}

Based on Laplacian energy,  the Laplacian centrality of a given vertex is defined as follows.

\begin{definition}
Let $G$ be a weighted graph and $G_i$ be the network obtained by deleting the vertex $v_i$ and its adjacent edges from $G$. Then, the Laplacian centrality $C_L(v_i,G)$ of $v_i$ is given by 
\[
\displaystyle C_L(v_i,G)= \frac{E_L(G)-E_L(G_i)}{E_L(G)}.
\]
\end{definition}

We prove that we can extend this centrality to the sheaf Laplacian as well.

\begin{theorem}
Let $G=(V,E)$ be a graph and $\mathcal{F}$ be a sheaf defined on $G$. Then, for a vertex $v_i \in V$, the Laplacian centrality based on the sheaf Laplacian, $C_{L_{\mathcal{F}}}(v_i,G)$ is well-defined. 
\end{theorem}
\begin{proof}
Let $L$ and $L_{\mathcal{F}}$ be the graph Laplacian and the sheaf Laplacian of $G$, respectively. There are two basic differences between these two matrices. First, the off-diagonal entries of $L$ are all nonpositive where $L_{\mathcal{F}}$ may have positive off-diagonal entry. Second, the sum of the off-diagonal entry in a row in $L$ equals to the negative of the diagonal entry on that row, but this is not necessarily true for $L_\mathcal{F}$. Here, we prove that these two properties of $L_\mathcal{F}$ do not have an effect on defining the Laplacian energy.

In Theorem 1 of \cite{qi2012laplacian}, they show that 
\[
\displaystyle E_L(G)=\sum_{i=1}^n d_i^2 + 2\sum_{i<j}w_{ij}^2
\]
where $d_i$ and $w_{ij}$ are on and off diagonal of $L$, respectively. As we see in this definition, we take the square of the off-diagonal entries, i.e., sign of these entries has no importance. This addresses the first difference. Moreover, in the proof of Theorem 1, they do not use the fact that $d_i=-\sum_{j, j\neq i}^n w_{ij}$, i.e., this difference is again no importance. This addresses the second difference. As a result, we can extend Laplacian centrality to the sheaf Laplacian.
\end{proof}

\subsubsection{DFF centrality}
DFF (diffusion Frechet function) centrality is based on the heat diffusion on networks. It is defined as the weighted sum of the diffusion distance between a vertex and the rest of the network, where the diffusion distance measures the similarity between given two nodes by finding the similarity of the heat diffusion on a given time interval when the heat source is located on these nodes. A more central vertex would have a similar heat diffusion with many vertices in the network, and as a result, it has a smaller DFF value. Mathematically, it is calculated as follows: Let $\mathcal{E}=[\mathcal{E}_1,..., \mathcal{E}_n]^T \in \mathbb{R}^n$ be a probability distribution on vertices of the graph $G=(V,E)$. For $t>0$, the diffusion Fr{\'e}chet function on a vertex $v_i \in V$ is defined as 
\begin{equation}\label{equ:dff}
F_{\mathcal{E},t}(i)=\sum_{j=1}^n d_t^2(i,j)\mathcal{E}_j
\end{equation}
with 
$$
d_t^2(i,j)=\sum_{k=1}^n e^{-2\lambda_kt}(\phi_k(i) - \phi_k(j))^2
$$
where $0\leq \lambda_1 \leq ... \leq \lambda_n$ are the eigenvalues of the graph Laplacian $L$ with orthonormal eigenvectors $\phi_1,...,\phi_n$. 

We now prove that we can extend DFF centrality to the sheaf Laplacian as well.

\begin{theorem}
Let $G=(V,E)$ be a graph and $\mathcal{F}$ be a sheaf defined on $G$. Then, for a vertex $v_i \in V$, the DFF centrality based on the sheaf Laplacian, $C_{DFF_{\mathcal{F}}}(v_i,G)$ is well-defined.
\end{theorem}

\begin{proof}
Let $L_{\mathcal{F}}$ be the sheaf Laplacian of $G$. In order to use Equation \ref{equ:dff} to define DFF for the sheaf Laplacian, $L_{\mathcal{F}}$ needs to have nonnegative real eigenvalues with orthonormal eigenvectors. But, since $L_{\mathcal{F}}$ is a positive-semidefinite symmetric matrix, it satisfies this condition. Hence, we can replace $L$ in Equation \ref{equ:dff} with $L_{\mathcal{F}}$ and define DFF centrality vertices of $G$.
\end{proof}

Therefore, we can use Laplacian centrality and DFF centrality via sheaf Laplacian to detect the influential nodes when deception is present in the social network.

\section{Experiments} \label{sec:exp}
In this section, we first introduce the evaluation methods and datasets we use in our experiments. We then present the results for the Laplacian and DFF centralities with different honesty parameter values ($\tau$) and determine the most effective relation type (honest, prosocial liar, and antisocial liar) on synthetic and real-world networks. 
\subsection{Evaluation} \label{sec:eval}
For each user $i$ in a network, we randomly select an opinion about a topic, $x_i$, within the interval $[-1,1]$. As the next step, we randomly divide the vertices into three equal parts and label each part as honest, prosocial liar, and antisocial liar. Then, we analyze the centrality scores of each label. There is a possible issue here that when we randomly divide the network into three parts, influential nodes when the deception is not present may accumulate in one of the parts. To avoid this, we first rank vertices from the most influential to the least for each centrality and divide them into 10 equal parts, i.e., the first part includes the top 10\% influential interactions and the last part includes the bottom 10\% influential vertices. We then randomly divide each part into three and label them as honest, prosocial liar and antisocial liar.

After assigning opinions and relation types (honest, prosocial liar, and antisocial liar) to each user, we define the linear transformation and the coboundary map. These provide the sheaf Laplacian. The sheaf Laplacian is also dependent on the honesty parameter $\tau \in [0,1]$. To see the effect of the different honesty levels, we partition the interval $[0,1]$ into 40 equal parts and input each value in the sheaf Laplacian. The experiment is run 100 times for each dataset, and the average of the 100 trials is taken to obtain more reliable results. 

As the final evaluation step, to get the influentiality score $S_R$ of a relation type $R$, we obtain the rankings based on each centrality in each simulation, take the average of the rankings for each relation type, and normalize it with the number of vertices in the network. In other words, the influentiality score is obtained by
\[
\displaystyle S_R=\frac{1}{|V|} \sum_{i=1}^{N} \sum_j^{|V_R|} c_{ij}
\]
where $|V_R|$ is the number of vertices of a given relation type, $N$ is the number of runs, and $c_{ij}$ is the centrality score of the $j$th vertex in $V_R$ in the $i$th run. Hence, the larger $S_R$ means a more influential relation in Laplacian centrality and a less influential relation in DFF centrality. 

\subsection{Datasets}
We consider nine real-world undirected social networks~\cite{nr,konect}, which have been widely adopted in the studies of influential node detection. (1) Train, a network containing contacts between suspected terrorists involved in the train bombing of Madrid on March 11, 2004 as reconstructed from newspapers. (2) Highschool, a network that contains friendships between boys in a small highschool in Illinois. (3) Lesmis, a network that contains co-occurances of characters in Victor Hugo's novel `Les Miserables'. (4) Copper, a network contains a common noun and adjective adjacencies for the novel David Copperfield by English 19th century writer Charles Dickens. (5) Jazz, a collaboration network between Jazz musicians. (6) Oz, a network contains friendship ratings between 217 residents living at a residence hall located on the Australian National University campus. (7) Congress, a network where nodes are politicians speaking in the United States Congress, and an edge denotes that a speaker mentions another speaker. (8) Innovation, a network spread among 246 physicians in five towns, i.e., Illinois, Peoria, Bloomington, Quincy, and Galesburg. (9) Netscience, a network of co-authorships in the area of network science. The general statistics of the datasets used for experiments are reported in Table \ref{table:data}. 

\begin{table}[h!]
\centering
\caption{Basic properties of the real-world datasets we use are provided here. $\langle k\rangle$ is the average degree, and $k_{max}$ is the maximum degree.}

\begin{tabular}{|c|c|c|c|c|c|}
\hline  $Dataset$ & $Vertices$ & $Edges$ & $\langle k\rangle$ & $k_{max}$  \\ \hline
\hline Train & 64 & 243 & 7.59 & 29  \\
\hline Highschool & 70 & 366 & 10.46 & 23  \\
\hline Lesmis & 77 & 254 & 6.60 & 36  \\
\hline Copper & 112 & 425 & 7.59 & 49  \\
\hline Jazz & 198 & 2742 & 27.69 & 100  \\
\hline Oz & 217 & 1839 & 16.94 & 56 \\
\hline Congress & 219 & 764 & 6.97 & 50 \\
\hline Innovation & 244 & 925 & 7.58 & 29 \\
\hline Netscience & 379 & 914 & 4.00 & 34 \\
\hline
\end{tabular}

\label{table:data}
\end{table}

\begin{figure*}[ht!]
\centering
\includegraphics[width=\textwidth]{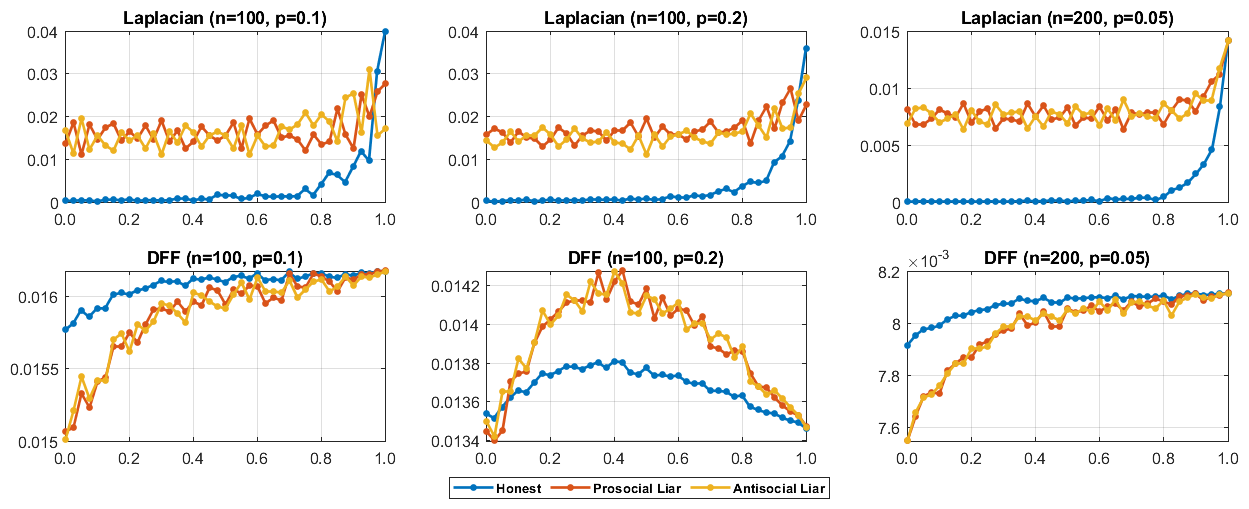}

\caption{The average centrality scores (y-axis) with respect to the different honesty parameter $\tau$ (x-axis) for Laplacian centrality and DFF centrality on three different Erdos-Renyi random graphs. Here, $n$ is the number of vertices in the random graph and $p$ is the probability of edge creation. While the bigger Laplacian centrality score implies being more influential, the smaller DFF centrality score implies being more influential.}

\label{fig:ER}
\end{figure*}

\begin{figure}[h!]
\centering
    \begin{subfigure}[t]{0.49\textwidth}
        \centering
\includegraphics[width=.9\textwidth]{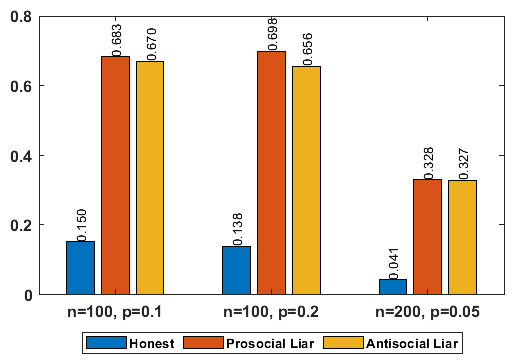}

\caption{ }
\end{subfigure}
\begin{subfigure}[t]{0.49\textwidth}
        \centering
\includegraphics[width=.9\textwidth]{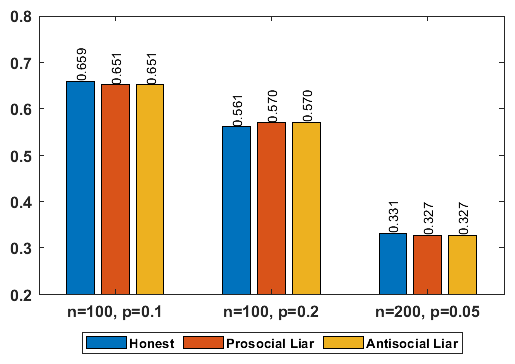}

\caption{ }
\end{subfigure}

\caption{The total centrality scores for Laplacian (a) and DFF (b) centralities of the Erdos-Renyi random graphs in Figure \ref{fig:ER}. }
\label{fig:ER_dff}
\end{figure}

\subsection{Results}
In this section, we present our results on different Erdos-Renyi random graphs and nine networks in Table \ref{table:data} with two different centrality methods: Laplacian centrality and DFF centrality. We start discussing the results for random graphs.
\begin{figure*}[t!]
\centering
\includegraphics[width=\textwidth]{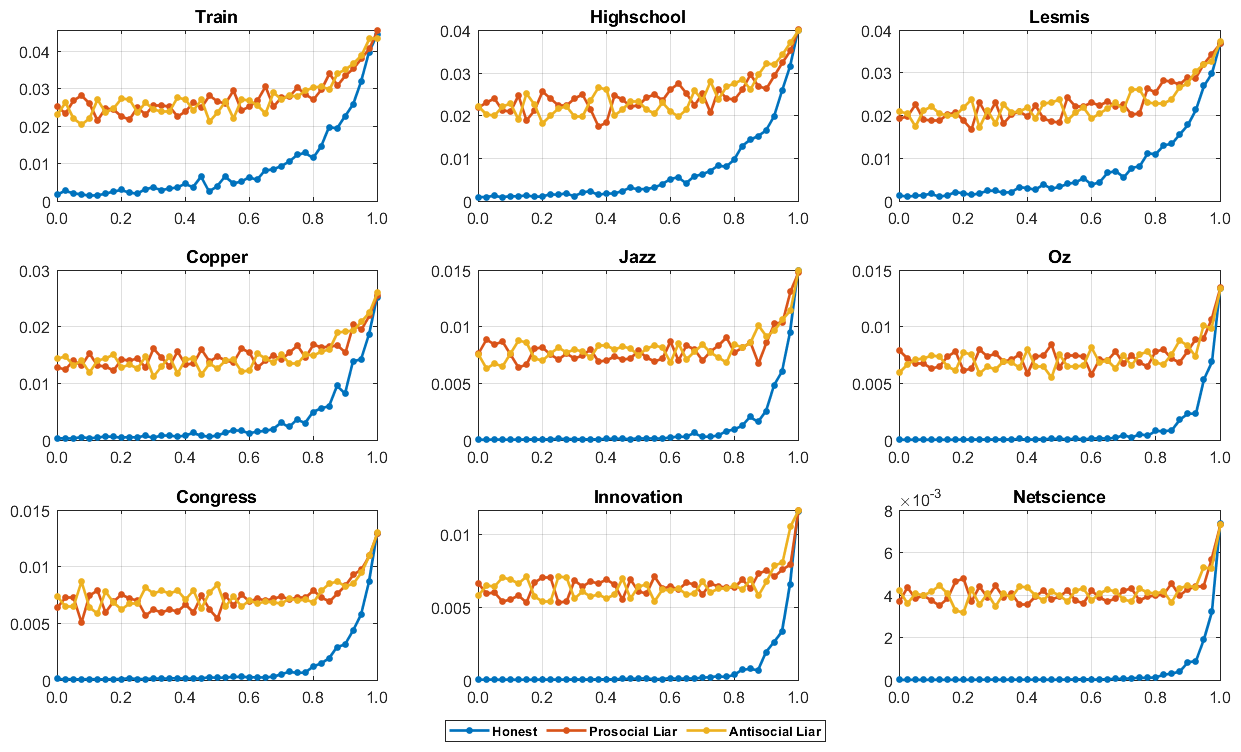}

\caption{The average centrality scores (y-axis) with respect to the different honesty parameter $\tau$ (x-axis) for Laplacian centrality on real social networks. The bigger Laplacian centrality score implies being more influential.}
\label{fig:LapR}
\end{figure*}

Erdos-Renyi random graphs take two parameters: the number of nodes ($n$) and the probability for edge creation $p$. To address different cases, we study three random graphs: (1) $n=100, p=0.1$, (2) $n=100, p=0.2$ and (3) $n=200, p=0.05$. If we take the first random graph as the test case, the second random graph has more density (i.e., the average degree) with keeping the size ($n$) the same and the third random graph is larger with keeping the density the same. Our goal here is to see how influentiality changes based on the density and the size of a network. Then, we evaluate the performance of each relation type (honest, prosocial liar, antisocial liar) on each network by following the outline in Section \ref{sec:eval}. The results are available in Figures \ref{fig:ER} and \ref{fig:ER_dff}.

As we see in Figure \ref{fig:ER}, for the Laplacian centrality, liars, regardless of being prosocial and antisocial, have bigger centrality scores than honest users on average for all random graphs. As we explain in the previous section, the bigger Laplacian centrality score implies being more influential, so liars are more influential than honest users for this centrality. Hence, density and size do not change this conclusion. For the DFF centrality, this conclusion is slightly different. As we see in the figure for DFF ($n=100, p=0.2)$, when the network becomes denser (the average degree is about 20 for this case), the honest users have smaller DFF centrality scores, i.e., are more influential. We better remind here that the smaller DFF centrality score implies being more influential as opposed to the Laplacian centrality. On the other hand, When the network is less dense, again the liars are more influential. The size of a network again does not change the result. Overall, we can easily conclude that liars are more influential than honest users on random graphs.

Furthermore, in Figures \ref{fig:ER_dff}, we present the total centrality scores for each random graph with Laplacian and DFF centralities, respectively. As we see in the figure for Laplacian centrality, the honest users are the least influential, and between prosocial and antisocial liars, prosocial ones are slightly more influential. In the figure for DFF centrality, liars are more influential when the average degree is smaller (i.e., for $n=100, p=0.1$ and $n=200, p=0.05$) whereas the honest users are more influential when the average degree is bigger. Hence, the edge density of the networks is an important factor for influence maximization when deception is present.   

\begin{figure*}[t!]
     \includegraphics[width=\textwidth]{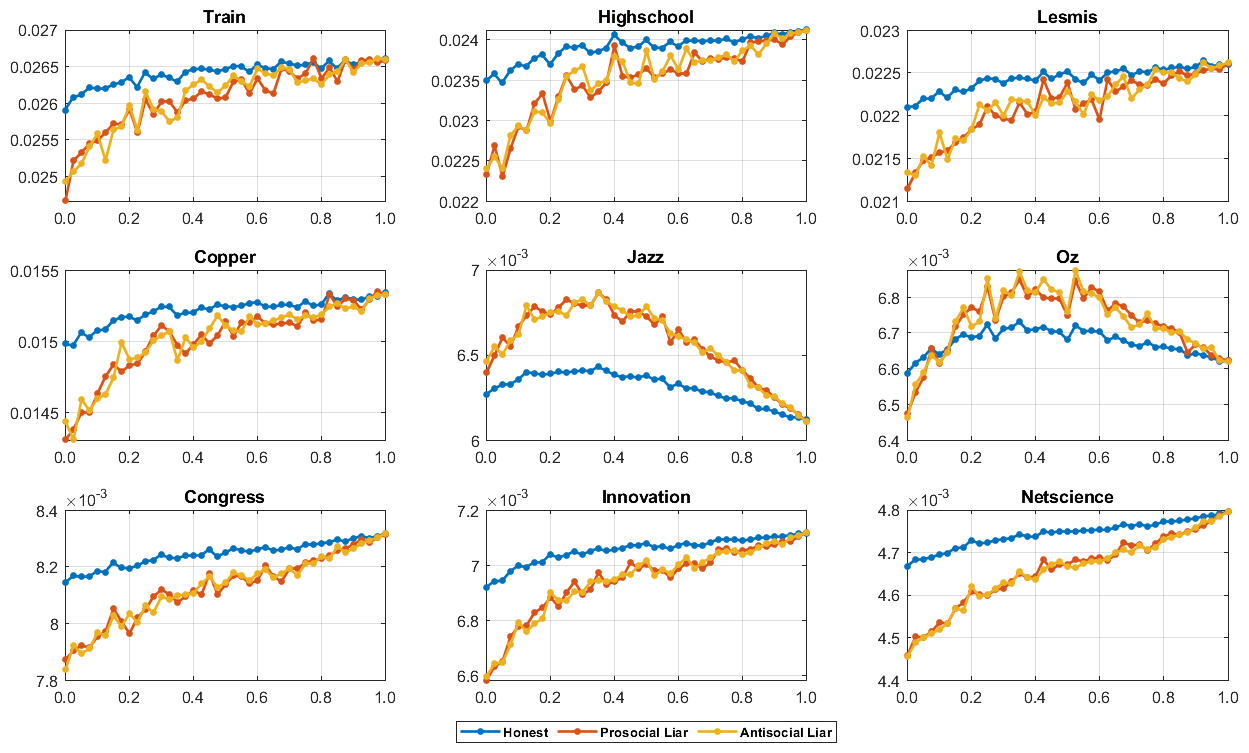}
    \caption{The average centrality scores (y-axis) with respect to the different honesty parameter $\tau$ (x-axis) for DFF centrality on real social networks. The smaller DFF centrality score implies being more influential.} 
\label{fig:DFFR}
\end{figure*}

Secondly, we present the simulation results on real social networks. We start with the Laplacian centrality where the results are available in Figure \ref{fig:LapR}. As we clearly see in the figure, liars, regardless of being prosocial and antisocial, again have bigger Laplacian centrality scores than honest users on average for all datasets. The difference gets smaller with a higher honesty level $\tau$ for each dataset. 

Furthermore, we observe that the Laplacian centrality score difference gets smaller when the network size gets bigger. For example, when $\tau=0$, i.e., when the total dishonest is present, the score difference for the Train network is about $0.025$. On the other hand, it is $0.004$ for the Netscience network. What we can conclude from here is that lying gets less important in terms of influence maximization when there are more vertices in the network.

Another interesting finding in this experiment is that lying makes users more influential regardless of being a prosocial or antisocial liar. Moreover, the total centrality score gets bigger as the honesty parameter $\tau$ gets closer to 1. 

Secondly, we use DFF centrality to measure the effect of deception in influence maximization where the results are available in Figure \ref{fig:DFFR}. As we see in the figure, the liars have a smaller average of DFF scores for all datasets but Jazz and Oz. This implies that they are again more influential than honest users. It is interesting that this pattern does not hold for Jazz and Oz networks. For these networks, honest users are more influential for most of the honesty parameter values. The common feature of these two networks is being denser than other networks. As we see in Table \ref{table:data}, the average degree for these networks are higher than other networks. Hence, we can conclude here that, based on this centrality, honest people become more influential whenever the network density is higher. We did not see this pattern for the Laplacian centrality. Overall, we can conclude based on the experiments of real social networks, liars are more influential than honest users. These results are also aligned with the results for the random graphs in Figures \ref{fig:ER} and \ref{fig:ER_dff}.

\begin{table}

\centering{
\caption{The total centrality scores for Laplacian and DFF centralities in Figure \ref{fig:LapR} and Figure \ref{fig:DFFR}. The darker cells correspond the more influential relation type.}
\label{table:comp}
\begin{tabular}{|c | c c c | c c c | }
 \hline

    \multirow{2}{*}{$ $} & \multicolumn{3}{|c|}{\cellcolor{gray!60}\textbf{Laplacian}} & \multicolumn{3}{|c|}{\cellcolor{gray!60}\textbf{DFF}} \\ \cline{2-7}
 & \cellcolor{gray!25}\textbf{Honest} & \cellcolor{gray!25}\textbf{Prosocial}& \cellcolor{gray!25}\textbf{Antisocial}  & \cellcolor{gray!25}\textbf{Honest} & \cellcolor{gray!25}\textbf{Prosocial}& \cellcolor{gray!25}\textbf{Antisocial} \\ \hline
 
 \cellcolor{gray!25} Train & 0.380 & \textbf{1.138} & 1.131 & 1.083 & \textbf{1.069} & 1.070\\ \hline
 
  \cellcolor{gray!25}Highschool & 0.279 & \textbf{1.012} & 1.005 & 0.980 & \textbf{0.964} & 0.965 \\ \hline
   \cellcolor{gray!25}Lesmis & 0.298 & 0.936 & \textbf{0.938} & 0.920 & \textbf{0.907} & \textbf{0.907} \\ \hline
   \cellcolor{gray!25}Copper & 0.141 & \textbf{0.623} & 0.608 & 0.623 & \textbf{0.615} & \textbf{0.615} \\ \hline
   \cellcolor{gray!25}Jazz & 0.049 & 0.333 & \textbf{0.338} & \textbf{0.259} & 0.270 & 0.270 \\ \hline
   \cellcolor{gray!25}Oz & 0.038 & \textbf{0.306} & 0.299 & \textbf{0.274} & 0.276 & 0.276\\ \hline
\cellcolor{gray!25}Congress & 0.048 & 0.297 & \textbf{0.310} & 0.338 & \textbf{0.331} & \textbf{0.331} \\ \hline
\cellcolor{gray!25}Innovation & 0.031 & \textbf{0.271} & 0.270 & 0.289 & \textbf{0.284} & \textbf{0.284} \\ \hline
\cellcolor{gray!25}Netscience & 0.016 & \textbf{0.170} & \textbf{0.170} & 0.195 & \textbf{0.191} & \textbf{0.191} \\ \hline \hline
\cellcolor{gray!25}Average & 0.142 & \textbf{0.565} & 0.563 & 0.551 & \textbf{0.545} & \textbf{0.545}\\ \hline

    \end{tabular}}
 \vspace{-0.35cm}
\end{table}

Besides the figures, we also present the total centrality score for each real social network in Table \ref{table:comp}. As we see in the table, based on the Laplacian centrality, prosocial and antisocial liars are the most influential interchangeably and on average, prosocial liars are slightly more influential than antisocial liars. For the DFF centrality, honest users are more influential only on Jazz and Oz networks, and for all other networks, prosocial and antisocial liars are again the most influential interchangeably and the same on average.  

\section{Conclusion} \label{sec:conc}
In this paper, we study the effect of deception in the influence maximization problem on social networks. We develop a method to model deception in social networks, employ the sheaf Laplacian to model the opinion dynamics when deception is present, and extend two node centrality measures, Laplacian centrality and DFF centrality, for the sheaf Laplacian to detect influential nodes. Our results show that liars are more influential than honest people in social networks. As future tasks, we plan to apply our method to understand the construction and circulation of truth in social media when deception is present. 

\bibliographystyle{unsrt}  
\bibliography{references,asli}

\end{document}